\documentclass[11pt,letterpaper]{article}
\usepackage{dblfloatfix}
\usepackage{bbm}
\usepackage{blindtext}
\usepackage{algpseudocode}
\usepackage{amsmath}
\usepackage{amssymb}
\usepackage{amsthm}
\usepackage{color}
\usepackage{cite}
\usepackage{graphicx}
\graphicspath{{./}{Figs/}}
\usepackage{subcaption}
\usepackage{mathrsfs}
\usepackage{verbatim}
\usepackage{indentfirst}
\usepackage{url}
\usepackage{epsfig,endnotes,amsthm,changepage}
\usepackage{authblk}

\usepackage{epstopdf,lipsum,etoolbox}

\usepackage{graphicx}
\usepackage{amsfonts}
\usepackage{setspace}
\usepackage{cite}
\usepackage{array}
\usepackage{mdwmath}
\usepackage{mdwtab}
\usepackage{mdwtab, stmaryrd}
\usepackage{times}
\usepackage{epsfig}
\usepackage{latexsym}
\usepackage{epstopdf}
\usepackage{verbatim}
\usepackage{units}
\usepackage{amsthm}
\usepackage{mdwlist}
\usepackage{dblfloatfix}
\usepackage{blindtext}
\usepackage{array}
\newcolumntype{P}[1]{>{\centering\arraybackslash}p{#1}}
\usepackage{tabu}

\AppendGraphicsExtensions{.pdf}
\usepackage{epstopdf,lipsum,etoolbox}

\newtheorem{theorem}{Theorem}

\newtheorem*{proposition*}{Proposition}

\newtheorem{lemma}{Lemma}
\newtheorem{remark}{Remark}
\newtheorem{example}{Example}

\newcommand\blfootnote[1]{%
  \begingroup
  \renewcommand\thefootnote{}\footnote{#1}%
  \addtocounter{footnote}{-1}%
  \endgroup
}

\allowdisplaybreaks

\begin{document}

\title{Latent-variable Private Information Retrieval}

 \author{{Islam Samy \quad Mohamed A. Attia \quad Ravi Tandon \quad Loukas Lazos 
 }\\
 Department of Electrical and Computer Engineering,  University of Arizona\\
 Email: \{\textit{islamsamy, madel, tandonr, llazos}\}@email.arizona.edu
 }


\maketitle
\blfootnote{The work of I. Samy and L. Lazos was supported by NSF grant CNS 1813401. The work of M. Attia and R. Tandon was supported by NSF Grants CAREER 1651492, CNS 1715947, and the 2018 Keysight Early Career Professor Award.}

\begin{abstract}
   In many applications, content accessed by users (movies, videos, news articles, etc.) can leak sensitive \textit{latent} attributes, such as religious and political views, sexual orientation, ethnicity, gender, and others. To prevent such information leakage, the goal of classical PIR is to hide the identity of the content/message being accessed, which subsequently also hides the latent attributes. This solution, while private, can be too costly, particularly, when perfect (information-theoretic) privacy constraints are imposed. For instance, for a single database holding $K$ messages, privately retrieving one message is possible if and only if the user downloads the entire database of $K$ messages. Retrieving content privately, however, may not be necessary to perfectly hide the latent attributes.  
  
 Motivated by the above, we formulate and study the problem of latent-variable private information retrieval (LV-PIR), which aims at allowing the user efficiently retrieve one out of $K$ messages (indexed by $\theta$) without revealing any information about the latent variable (modeled by $S$). We focus on the practically relevant setting of  a single database and show that one can significantly reduce the download cost of LV-PIR (compared to the classical PIR) based on the correlation between $\theta$ and $S$. 
  We present a general scheme for LV-PIR as a function of the statistical relationship between $\theta$ and $S$, and also provide new results on the capacity/download cost of LV-PIR. Several open problems and new directions are also discussed. 
  
\end{abstract}

\section{Introduction}

Consider the following scenario. Alice accesses an online service to obtain information hosted on a remote database. The retrieved information reveals with high certainty Alice's political affiliation, making her a target of political advertisements and biased news reports. This scenario is commonplace for the plethora of online services that employ data analytics to profile users and infer private information such as their religious and  political views, sexual orientation, ethnicity, gender, health state, and personality traits \cite{kosinski2013private}. The risks for potential misuse of personal information are further exacerbated by frequent data breaches and flawed privacy policies that lead to privacy leakages such as the infamous Facebook-Cambridge-Analytica incident in 2015 \cite{cadwalladr2018revealed}.

The default technical approach to preserve privacy when data is retrieved from one or more databases is to allow for private information retrieval (PIR) \cite{chor1995private}. By employing a PIR scheme, a user can download intended messages without revealing the message indices to the database. The PIR problem was initially introduced assuming limits on the computational capabilities of the databases \cite{ostrovsky2007survey}. Given the ever-increasing computational power of high-performance computing,  approaches that achieve information-theoretic privacy have recently gained significant attention \cite{shah2014one,sun2017capacity,sun2017optimal,sun2019capacity,
sun2018multiround,sun2018capacity,tajeddine2018private,banawan2018capacity,
wang2017symmetric,freij2017private,sun2018private,jia2019x,lin2018mds,
sun2019breaking,zhou2019capacity,banawan2018multi,zhang2017private,
banawan2019capacity,tajeddine2018robust,banawan2020private,wang2018secure,
wang2018capacity,tandon2017capacity,wei2018fundamental,shariatpanahi2018multi,
attia2018capacity,samy2019capacity,tian2019capacity,wang2019symmetric,
tajeddine2019private,yang2018private,jia2019cross,kumar2019achieving,
raviv2018private,banawan2018private,banawan2020capacity} for different model setups. These approaches investigate the PIR capacity, defined as the maximum ratio of the desired data to the total amount of downloaded data. 
However, when a single database is considered or databases belong to a single operator as is the case for the majority of online services, information-theoretic PIR methods degenerate to communication-expensive solutions requiring the download of the entire database to guarantee privacy.  

In this paper, we investigate if the download cost for the single database scenario can be reduced, while still protecting user privacy. We argue that the content queried by the user is often an intermediate data product that is exploited to infer latent data traits. We focus on providing information-theoretic privacy for the latent traits, even if the intermediate data is exposed. We model the user profile with a {\em latent variable model} captured by a latent random variable $S$. 



\begin{figure}[t]
\begin{center}

\includegraphics[width=0.7\linewidth,=1]{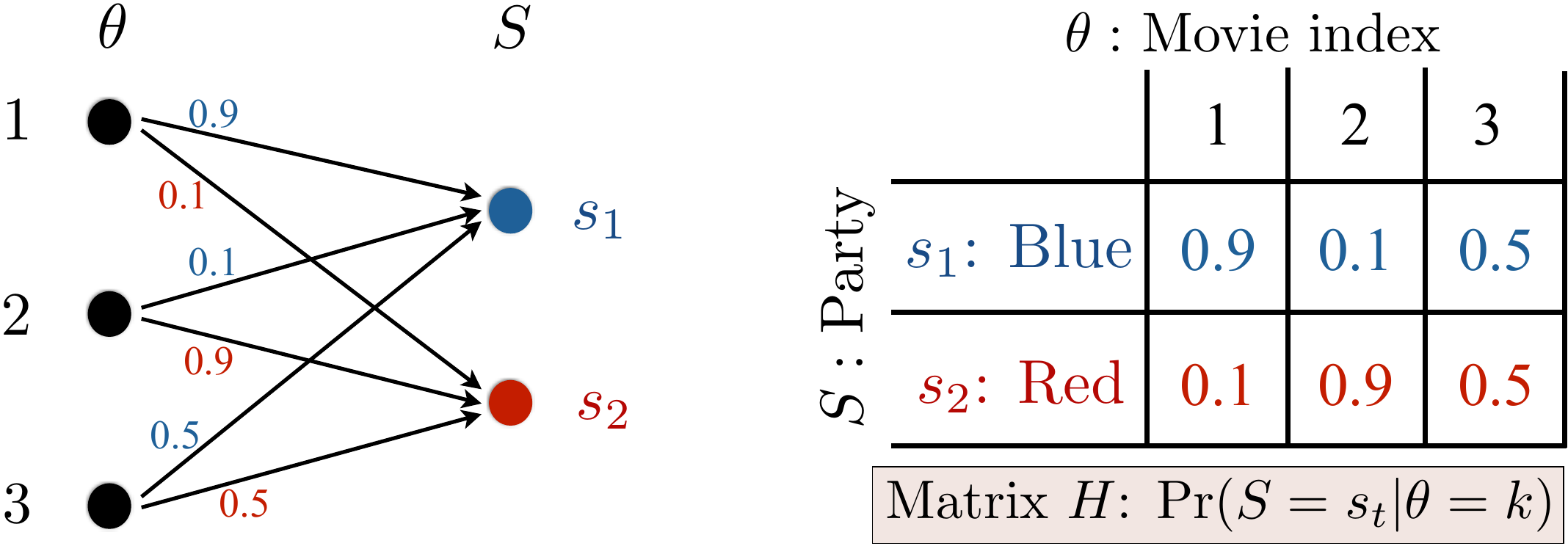}
\end{center}
\caption{Sample realization of the characteristic matrix $\mathbf{H}$. The random variable $\theta$ represents the message index whereas $S$ represents the latent variable.} 
\label{fig:Ex-party}
\end{figure}

To provide an example, consider that Alice uses a streaming service hosting three movies labeled by index $\theta \in\{1,2,3\}.$ The movie access patterns are analyzed to infer if Alice leans towards the blue or the red party. Here, the  latent variable $S$ takes only two values, namely red and blue.
The relationship between the content and the affiliation is captured by the conditional probability distribution $\Pr(S\vert \theta),$ expressed as a matrix $\mathbf{H}$ of size $2 \times 3$, with the $(t,k)$ element being $h_{tk} = \Pr(S = s_t \vert \theta = k)$. A sample realization of  $\mathbf{H}$ is shown in Fig~\ref{fig:Ex-party}. Evidently, if Alice retrieves all three movies, she is associated with each affiliation equiprobably (the marginal distribution of $S$ is uniform). However, the same privacy constraint can be satisfied with a lower download cost. If Alice wants to retrieve Movie $3,$ then querying set $\{3\}$ is sufficient to hide Alice's political affiliation. Although, under this query strategy, the message index is leaked to the database, the index reveals  no {\em additional information} with respect to the latent variable $S,$ which remains uniformly distributed. If Alice wants to retrieve either Movie $1$ or Movie $2,$ then querying set $\{1, 2 \}$ yields a uniform marginal distribution on $S$.  This toy example demonstrates two important points: (a) by considering the privacy of the latent variable, the download cost of PIR can be reduced even for a single database, and (b) the cost becomes dependent on the intended message $k$ and the conditional distribution $\Pr(S\vert \theta)$ captured by $\mathbf{H}.$ 



In this paper, we investigate the advantages of a privacy model that focuses on protecting latent variables as opposed to the data retrieval patterns. We formulate the latent-variable PIR (LV-PIR) problem under a simple, but practically relevant setting of one uncoded database and  provide relevant correctness and privacy definitions. We derive privacy constraints to meet those privacy definitions and formulate a general optimization problem for minimizing the download cost of the LV-PIR. We show that the optimal cost is dependent on the structure and values of $\mathbf{H}.$ In the extreme case where $\mathbf{H}$ is full column rank (each message imposes a unique distribution on $S$), we show that no improvements can be achieved relevant to the classic PIR definition. For the more reasonable case where groups of messages yield the same distribution on $S$, we construct a query strategy based on the structure of $\mathbf{H}$ which reduces the average download cost by a constant factor. Finding the optimal LV-PIR strategy for any $\mathbf{H}$ is left as an open problem.


\section{Problem Formulation}
\label{Sec:prob-form}

\begin{figure}[t]
\begin{center}
\includegraphics[width=0.7\linewidth]{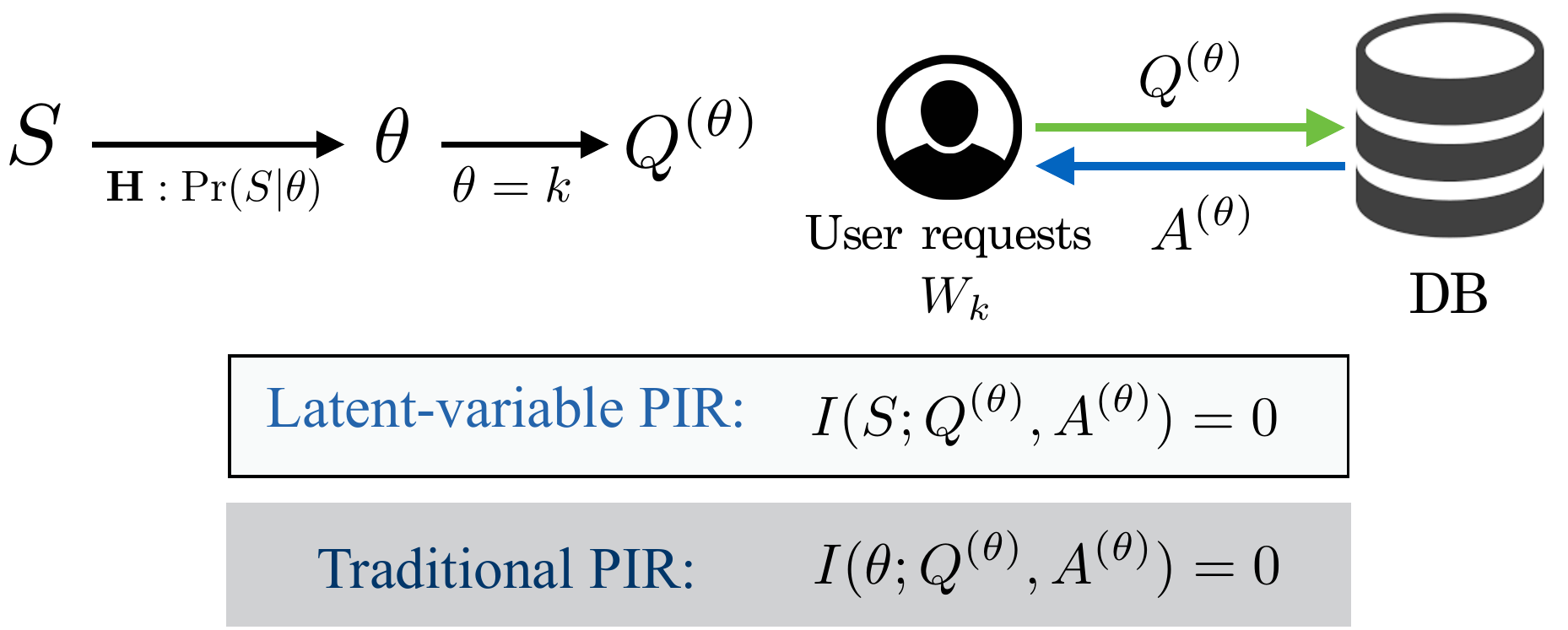}
\end{center}
\caption{\label{fig:sysmod} Latent-variable PIR (LV-PIR) setting for one DB.}
\end{figure}


We consider the PIR setting in Fig.~\ref{fig:sysmod} where a set $\mathcal{W} =\{W_1,W_2, \dots, W_K\}$ of $K$ independent messages, each of size $L$, are stored on one database (DB). 
 The user draws a message index $\theta\sim \vec{p}$, where $\vec{p}$ denotes the $K\times 1$ vector containing the probabilities of selecting the $K$ messages, i.e., $\vec{p}= [\Pr(\theta=1)\ \Pr(\theta=1)\ \cdots \ \Pr(\theta=K)]^T$. The user is interested in retrieving message $W_{\theta}$ while hiding a latent variable $S$, which can take one of  $T$ values from alphabet $\mathcal{S}=\{s_1, s_2,\dots, s_T\}$.
The variable $S$ is dependent on the message index $\theta$, as described by the conditional probability $\Pr(S=s_t|\theta=k)$  where  $t\in[1:T]$ when the desired message is $W_{\theta=k}$, $k\in[1:K]$.
We capture this conditional probability via a $T\times K$ matrix $\mathbf{H}$  with $[\mathbf{H}]_{t,k}\triangleq h_{tk}= \Pr(S=s_t|\theta=k)$. 

We focus on the case where the matrix $\mathbf{H}$ is fixed and publicly-known to the DB and the user. Let  $\overrightarrow{\mathbf{H}(t)}$ be the $t^{th}$ row of $\mathbf{H}$ and $\vec{p}$ be the message probability column vector. The prior distribution of $S$ can be obtained from $\mathbf{H}$ and $\vec{p}$ as 
\begin{equation}
\label{eq:S_init}
   \Pr({S}=s_t)=\sum\limits_{k=1}^{K}\Pr(\theta=k)\cdot \Pr({S}=s_t|\theta=k)=\overrightarrow{\mathbf{H}(t)}\cdot \vec{p }.
\end{equation}


 To retrieve message $W_{\theta}$, the user submits a query $Q^{(\theta)}$ to the DB. The DB determines the corresponding answer  $A^{(\theta)}$ as a function of the query $Q^{(\theta)}$ and the $K$ messages. Therefore,
\begin{equation}\label{answerfunction}
H(A^{(\theta)}
 |Q^{(\theta)},W_{1},\ldots,W_{K}) = 0.
\end{equation}
The variables $S\rightarrow \theta \rightarrow Q^{(\theta)}$ form  a Markov chain as $S$ is conditionally independent of  $Q^{(\theta)}$ given the  index $\theta$, i.e.,
\begin{align}
\label{eq:markov}
    \Pr(S=s_t|\theta = k,Q^{(\theta)} =q) = \Pr(S=s_t|\theta = k).
\end{align}

\noindent An LV-PIR scheme must satisfy the following correctness and privacy constraints.

\noindent $\bullet$ \hspace{5pt} \textbf{Correctness Constraint:}
The user must be able to recover the requested message $W_\theta$ from  $A^{(\theta)},$
\begin{equation}
\label{eq:corre}
H(W_{\theta}|Q^{(\theta)},A^{(\theta)})=0.
\end{equation}

\noindent $\bullet$ \hspace{5pt} \textbf{Latent-variable Privacy Constraint:} The latent variable $S$ should be independent of the submitted query $Q^{(\theta)}$ and its corresponding answer $A^{(\theta)}$, i.e, $Q^{(\theta)}$ and $A^{(\theta)}$ should not leak any \textit{additional} information about $S$ than what is known by the prior distribution of $S$. For any realization of $(Q^{(\theta)}=q,A^{(\theta)}=a)$, we must have
\begin{equation}\label{eq:5}
  \Pr({S}=s_t|Q^{(\theta)}=q,A^{(\theta)}=a) =\Pr({S}=s_t)=\overrightarrow{\mathbf{H}(t)}\cdot \vec{p}. 
\end{equation}
Equivalently, the condition in \eqref{eq:5} can be expressed as:
\begin{equation}
    I(S;Q^{(\theta)},A^{(\theta)})=0.
\end{equation}

 We focus on the case where all messages are equiprobable, i.e., $\vec{p} = \frac{1}{K} \vec{\mathbbm{1}}$, where $\vec{\mathbbm{1}}$ denotes the $K\times 1$ all-ones column vector. The privacy constraint in \eqref{eq:5} can be rewritten as: 
\begin{equation}
\label{eq:pri}
  \Pr({S}=s_t|Q^{(\theta)}=q,A^{(\theta)}=a) =\Pr({S}=s_t)=\frac{1}{K}\overrightarrow{\mathbf{H}(t)}\cdot \vec{\mathbbm{1} }.
\end{equation}

We evaluate the overhead of LV-PIR schemes using the download cost for recovering a desired message while achieving the latent-variable privacy constraint. However, since the structure of $\mathbf{H}$ might not be symmetric, the download cost becomes message-dependent. Therefore, we consider the average download cost over all messages.
Let  $D_{\mathbf{H}}(k)$ be the number of downloaded bits for a desired message $W_{k}$ and a characteristic matrix $\mathbf{H}$. The average number of downloaded bits is then $D_{\mathbf{H}} = \frac{1}{K} \sum_{k=1}^K D_{\mathbf{H}}(k)$. The pair $(L,D_{\mathbf{H}})$ is said to be achievable if there exists an LV-PIR scheme that satisfies the correctness  (eq. \ref{eq:corre}) and latent-variable privacy (eq. \ref{eq:pri}) constraints, and can retrieve a message of size $L$ bits by downloading an average of $D_{\mathbf{H}}$ bits. Our goal is to understand the optimal download cost $D^{*}(\mathbf{H})$, defined as 
\begin{equation}
    D^{*}(\mathbf{H})=\min\{{D_{\mathbf{H}}}/{L}:(L,D_{\mathbf{H}}) \  \text{is achievable}\}.
    \end{equation}

\section{Main Results}

As a baseline, the classical PIR scheme for one DB, i.e., downloading all $K$ messages also satisfies latent-variable privacy. This is
because the query of downloading all messages is independent of $\theta$, and by the Markov chain, $S\rightarrow \theta \rightarrow Q^{(\theta)}$, is also independent of $S$. Thus, the baseline scheme gives a trivial bound $D^{*}(\mathbf{H})\leq K$. 
Suppose now there is a scheme that only downloads a subset of the $K$ messages including the requested message. It is obvious that this scheme no longer satisfies message index privacy, however,  depending on  $\mathbf{H}$, the latent-variable privacy constraint in \eqref{eq:pri} could still be satisfied. To demonstrate this, let us revisit the motivating example in  Fig.~\ref{fig:Ex-party} in further detail.


\begin{example}
\normalfont
Consider three messages indexed by $\theta = \{1, 2, 3\}$ that are related to a latent variable $S$ taking two values $s_1$ and $s_2$. Let $\mathbf{H}$, expressing the relationship (conditional probability) between $S$ and $\theta$, be
\begin{align}
    \mathbf{H} = \left[\begin{array}{ccc}
          0.1 &0.9 &0.5 \\
          0.9 &0.1 &0.5
    \end{array}\right].\label{eq:Hexample}
\end{align}
From \eqref{eq:S_init} and \eqref{eq:Hexample},  the prior distribution of $S$ can be computed as 
 $
    \Pr(S=s_t)= 0.5, \forall t\in\{1,2\}.
$
Suppose that the desired message is indexed by $\theta = 3$. In this case,   requesting message $W_3$ alone leads to the same distribution over $S$ compared to the prior distribution.  
Suppose now that the desired message is indexed by $\theta=1$ (respectively, $\theta=2$). In this case, requesting message $W_1$ (respectively, $W_2$) alone implies a different posterior distribution on $S$ compared to the prior distribution. 
To match the prior distribution, whenever $\theta=1$ or $\theta=2$, the user can download both $\{W_1, W_2\}$. 
Based on this discussion, we have the following query structure for the specific $\mathbf{H}$: \begin{align}
    Q^{(\theta)} =\begin{cases}
    \{1,2\}, &\quad \theta = 1 \text{ or } \theta=2,\\
    \{3\}, &\quad \theta = 3.
    \end{cases}
\end{align}
It is easy to verify that the query structure satisfies the privacy constraint, i.e., $\Pr(S=s_t|Q^{(\theta)}=q) = \Pr(S=s_t) =0.5$, $\forall t, q$. The download cost for each message is $D_{\mathbf{H}}(3)=1$ and $D_{\mathbf{H}}(1)=D_{\mathbf{H}}(2)=2$. Consequently, the average download cost  $D_{\mathbf{H}}=\nicefrac{1}{3}\cdot(2+2+1) = \nicefrac{5}{3}<3$, which is lower than the download cost of the classical PIR.  

%
\end{example}
\subsection{Sufficient Conditions for Latent Variable Privacy}
In this section, we present sufficient conditions that any  LV-PIR scheme must satisfy for an arbitrary $\mathbf{H}$.
Suppose that the message index is $\theta=k$, i.e.,  the desired message is $W_{k}$. Let the  user send a subset query $Q^{(\theta)}=\mathcal{M}_k$  of size $|\mathcal{M}_{k}|\leq K$ to the DB, requesting to download all messages whose indices are in $\mathcal{M}_{k}$. The set $\mathcal{M}_{k}$ can also be expressed as a $K\times 1$ binary vector $\vec{b}_{\mathcal{M}_{k}}\in \{0,1\}^K$ with 
\begin{equation}
   \vec{b}_{\mathcal{M}_{k}}(i)=\begin{cases} 1 , \quad  \ i\in \mathcal{M}_{k},\\
    0 , \quad  \ i\not\in \mathcal{M}_{k}. 
    \end{cases}
\end{equation}
To satisfy the correctness constraint, the desired message $k$ must be downloaded, i.e., $k \in \mathcal{M}_{k}$.

We assume a class of query structures
such that the conditional distribution of $\theta$ given the query $Q^{(\theta)}$ is uniform and can be obtained as
\begin{equation}
    \Pr(\theta=i|Q^{(\theta)}=\mathcal{M}_k)=\begin{cases}\frac{1}{|\mathcal{M}_{k}|}, &  \ i\in \mathcal{M}_k,\\
    0, & \ i\not\in \mathcal{M}_k.
    \end{cases}\label{eq:conditionaltheta}
\end{equation}
In other words, given the query $Q^{(\theta)} = \mathcal{M}_k$, all the messages in $\mathcal{M}_k$ are equally likely to be requested from the database perspective.
In the next Lemma, we present a necessary condition for a subset query to satisfy the latent-variable privacy constraint \eqref{eq:pri}. Intuitively, this Lemma formalizes the privacy constraint, i.e., the posterior distribution  of $S$ given any query realization should match the prior distribution of $S$.

\begin{lemma}
\label{Thm1}
An LV-PIR scheme for one DB, aiming to download message $W_k$ with a query $Q^{(\theta)}=\mathcal{M}_{k}$, as described by the binary vector $\vec{b }_{\mathcal{M}_{k}}$,  is private for a latent variable $S$ if 
\begin{equation}
\label{eq:Thm1}
    \frac{1}{|\mathcal{M}_{k}|}{\mathbf{H}}\cdot \vec{b}_{\mathcal{M}_{k}}=\frac{1}{K}{\mathbf{H}}\cdot \vec{\mathbbm{1} }.
\end{equation}
\end{lemma}
\begin{proof}
The above result follows due to the following arguments: a) if $i\not\in \mathcal{M}_k$, then $i$ cannot be the requested message index as the correctness constraint will be violated; and b) if
$i\in \mathcal{M}_k$, then from DB's perspective, all the messages in $\mathcal{M}_{k}$ are equiprobable.
Given the above statements, we can now compute the conditional distribution of $S$  given $Q^{(\theta)}=\mathcal{M}_k$ (posterior of $S$ given the query) as follows,
\begin{align}
   &\Pr(S=s_t|Q^{(\theta)}=\mathcal{M}_k)\nonumber\\
   &=\sum_{i=1}^{K} \Pr(\theta=i|Q^{(\theta)}=\mathcal{M}_k)\cdot \Pr({S}=s_t|\theta=i,Q^{(\theta)}=\mathcal{M}_k)\nonumber\\
   &\overset{(a)}{=}\frac{1}{|\mathcal{M}_{k}|}\sum_{i\in\mathcal{M}_k}\Pr(S=s_t|\theta=i,Q^{(\theta)}=\mathcal{M}_k)\nonumber\\
  &\overset{(b)}{=}\frac{1}{|\mathcal{M}_{k}|}\sum_{i\in\mathcal{M}_k}\Pr({S}=s_t|\theta=i)=\frac{1}{|\mathcal{M}_{k}|}\overrightarrow{\mathbf{H}(t)}\cdot \vec{b}_{\mathcal{M}_{k}},\label{eq:sgivenquery}
   \end{align}
where $(a)$ follows from \eqref{eq:conditionaltheta} and $(b)$ follows from \eqref{eq:markov}.  Hence, using \eqref{eq:sgivenquery}, the query satisfies latent-variable privacy constraint in \eqref{eq:pri}, if 
    $\frac{1}{|\mathcal{M}_{k}|}\overrightarrow{\mathbf{H}(t)}\cdot \vec{b}_{\mathcal{M}_{k}}=\frac{1}{K}\overrightarrow{\mathbf{H}(t)}\cdot \vec{\mathbbm{1} }, \quad\forall t\in [1:T]$,
or equivalently, $   \frac{1}{|\mathcal{M}_{k}|}{\mathbf{H}}\cdot\vec{b}_{\mathcal{M}_{k}}=\frac{1}{K}{\mathbf{H}}\cdot \vec{\mathbbm{1} }$.
\end{proof}

\subsection{A General LV-PIR Scheme based on exhaustive search (ES)}
Using Lemma \ref{Thm1}, we next describe a general LV-PIR scheme. We assume $\ell$ disjoint partitions of the message indices $\mathcal{L}_1,\ldots,\mathcal{L}_{\ell}$ such that $\mathcal{L}_j\subseteq [1:K]$ and $\sum_{j=1}^\ell |\mathcal{L}_j|=K$. In order to download a message indexed by $\theta=k$, the user sends the query $Q^{(\theta)}=\mathcal{L}_j$ such that $k\in \mathcal{L}_j$. 
This query structure insures the condition \eqref{eq:conditionaltheta} is satisfied. That is for a given query $Q^{(\theta)}=\mathcal{L}_j$, all the messages $i\in \mathcal{L}_j$ are equally likely. The average number of downloaded bits for this LV-PIR scheme is given as follows:
\begin{align}
D_{\mathbf{H}} = \frac{1}{K}\sum_{k=1}^K D_{\mathbf{H}}(k) = \frac{1}{K} \sum_{j=1}^\ell \sum_{k\in \mathcal{L}_j} |\mathcal{L}_j|L = \frac{L}{K} \sum_{j=1}^\ell |\mathcal{L}_j|^2.
\end{align}

The key idea is that we can perform an exhaustive search over all possible partitions $\mathcal{L}_1,\ldots,\mathcal{L}_{\ell}$ such that the privacy condition in Lemma \ref{Thm1} is satisfied and that the average download cost is minimized. This scheme  leads to an upper bound on the optimal average download cost as stated in the following Theorem.

\begin{theorem}
\label{thmexhaustive}
The optimal average download cost for LV-PIR from one DB is upper bounded as follows
\begin{align}
    &D^*(\mathbf{H})\leq D^{\text{ES}}(\mathbf{H}) =\min_{\mathcal{L}_1,\ldots,\mathcal{L}_{\ell}} \frac{1}{K}\sum_{j=1}^{\ell} |\mathcal{L}_{j}|^2, \\
    \text{s.t.} 
    \quad 
    &\mathcal{L}_{j}\subseteq[1:K], \quad \cup_{j=1}^{\ell} \mathcal{L}_{j} =[1:K], \quad \mathcal{L}_{j}\cap \mathcal{L}_{j'} =\emptyset, \nonumber\\
    &\frac{1}{|\mathcal{L}_{j}|}{\mathbf{H}}\cdot \vec{b}_{\mathcal{L}_{j}}=\frac{1}{K}{\mathbf{H}}\cdot \vec{\mathbbm{1}}, \quad  \quad \forall j\neq j'\in[1:\ell].
\end{align}
\end{theorem}



In the next Theorem, we show that when the characteristic matrix $\mathbf{H}$ is full column rank, then the exhaustive search scheme becomes equivalent to classical PIR, i.e., one cannot do any better than downloading all $K$ messages for LV-PIR. 

\begin{theorem} 
\label{theoremfullrank}
If the characteristic matrix $\mathbf{H}$ is full column rank, i.e., the columns of $\mathbf{H}$ are independent, then $D^{\text{ES}}(\mathbf{H})=K$. In other words, the exhaustive search scheme is equivalent to the classical PIR scheme of downloading all $K$ messages. 
\end{theorem}
\begin{proof}
We prove this result by contradiction. Consider a matrix $\mathbf{H}$ with full column rank.  For a requested message index $\theta=k$, suppose there exists a scheme that downloads a subset  $Q^{(\theta)}=\mathcal{L}_j$ with $k\in \mathcal{L}_j$ and  $|\mathcal{L}_j|<K$ while satisfying  latent-variable privacy constraint in Lemma \ref{Thm1}: 
 \begin{equation}
 \label{eq:lem1-1}
    \frac{1}{| \mathcal{L}_j|}{\mathbf{H}}\cdot \vec{b }_{\mathcal{L}_j}=\frac{1}{K}{\mathbf{H}}\cdot \vec{\mathbbm{1} }.
\end{equation}
We can express this constraint as follows:
\begin{equation}
\label{eq:rem1-col}
    \frac{1}{| \mathcal{L}_j|}\sum_{i\in\mathcal{L}_j}\overrightarrow{\mathbf{H}(:,i)} =\frac{1}{K}\sum_{i=1}^{K}\overrightarrow{\mathbf{H}(:,i)},
\end{equation}
where $\overrightarrow{\mathbf{H}(:,i)}$ is the $i^{th}$ column of $\mathbf{H}$. Rearranging \eqref{eq:rem1-col}, we get
\begin{align}
    K\sum_{i\in\mathcal{L}_j}\overrightarrow{\mathbf{H}(:,i)} 
    &=| \mathcal{L}_j|\sum_{i\in \mathcal{L}_j}\overrightarrow{\mathbf{H}(:,i)} +| \mathcal{L}_j|\sum_{i\not\in \mathcal{L}_j}\overrightarrow{\mathbf{H}(:,i)}, \nonumber
\end{align}
which can be rearranged as follows:
\begin{equation}
\label{eq:rem1_25}
    (K-| \mathcal{L}_j|)\sum_{i\in\mathcal{L}_j}\overrightarrow{\mathbf{H}(:,i)} =| \mathcal{L}_j|\sum_{i\notin\mathcal{L}_j}\overrightarrow{\mathbf{H}(:,i)}.
\end{equation}
Let us consider any index $c\in\mathcal{L}_j$. Then from \eqref{eq:rem1_25}, we have
\begin{equation}
    \overrightarrow{\mathbf{H}(:,c)}=\frac{| \mathcal{L}_j|}{K-| \mathcal{L}_j|}\sum_{i\notin\mathcal{L}_j}\overrightarrow{\mathbf{H}(:,i)}-\sum_{i\in \mathcal{L}_j\setminus{\{c\}}}\overrightarrow{\mathbf{H}(:,i)}.
\end{equation}
This implies that the column $\overrightarrow{\mathbf{H}(:,c)}$ can be written as a linear combination of the remaining $K-1$ columns of $\mathbf{H}$, contradicting the fact that $\mathbf{H}$ has full column rank. Thus, there does not exist any $\mathcal{L}_j$ with $| \mathcal{L}_j|< K$ that satisfies \eqref{eq:Thm1}.
\end{proof}

The result of Theorem~\ref{theoremfullrank} may seem rather pessimistic indicating that when each message index ``independently'' maps to the latent variable, then the only option is to download the entire database.
However, in many practical scenarios, the number of messages $K$ (e.g., number of movies in a DB) is significantly larger than the alphabet size $|\mathcal{S}|=T$ (e.g., political affiliation of a user). In these cases, the matrix $\mathbf{H}$ is not full column rank and the message indices mapping to the latent variable are linearly dependent. 
When $\mathbf{H}$ is not full column rank, one can always resort to the exhaustive search based scheme presented in Theorem~\ref{thmexhaustive}. However, this scheme has exponential complexity and may scale poorly with the size of the DB. 
Next, we examine a special case for the characteristic matrix $\mathbf{H}$ for which we present a low-complexity LV-PIR scheme that can also reduce the download cost. 


\subsection{Low-complexity LV-PIR Scheme}

For practical cases when the number of messages $K$ is large (in comparison to $T$), we expect that groups of messages are statistically equivalent. For instance, a group of movies from the same director or the same studio may imply a similar political affiliation (conditional probability). In such cases, the characteristic matrix $\mathbf{H}$ can have several repeated  columns. This gives the opportunity to make equivalent selections of messages when constructing a query and potentially reduce the download cost. We first demonstrate this with an example. 

\begin{example}
\normalfont
Let us consider a LV-PIR setting with $K=6$ messages, and a latent variable $S$ taking $T=3$ values, correlated via a matrix $\mathbf{H}$ of the form 
\begin{equation}
\mathbf{H}=\begin{bmatrix}
    0.3 & 0.3&  0.3 & 0.3 & 0.4 & 0.4\\
    0.1 & 0.1&  0.1 & 0.1 & 0.3 & 0.3\\
    0.6 & 0.6&  0.6 & 0.6 & 0.3 & 0.3\\
\end{bmatrix}.
\end{equation}
We notice that this structure of $\mathbf{H}$ implies two groups of messages with the same mapping to the latent variable $S$, $\mathcal{G}_1=\{W_1,W_2,W_3,W_4\}$, and $\mathcal{G}_2=\{W_5,W_6\}$. Assume downloading subsets $\mathcal{G}'_1\subseteq \mathcal{G}_1$ and $\mathcal{G}'_2 \subseteq \mathcal{G}_2$ of messages from the two groups, with the inclusion of the desired message. For such a query structure, the privacy condition in Lemma~\ref{Thm1} can be expressed as
\begin{equation*}
  \frac{|\mathcal{G}'_1|}{|\mathcal{G}'_1|+|\mathcal{G}'_2|}\begin{bmatrix}
    0.3 \\
    0.1 \\
    0.6 \\
\end{bmatrix}+\frac{|\mathcal{G}'_1|}{|\mathcal{G}'_1|+|\mathcal{G}'_2|}\begin{bmatrix}
    0.4 \\
    0.3 \\
    0.3 \\
\end{bmatrix}=  \frac{4}{6}\begin{bmatrix}
    0.3 \\
    0.1 \\
    0.6 \\
\end{bmatrix}+\frac{2}{6}\begin{bmatrix}
    0.4 \\
    0.3 \\
    0.3 \\
\end{bmatrix}.
\end{equation*}
It is obvious that the choice $\mathcal{G}'_1=\mathcal{G}_1$ and $\mathcal{G}'_2=\mathcal{G}_2$ satisfies the above condition, which is the classical PIR scheme (downloading all messages).
Another choice that also satisfies the above constraint is to pick $\mathcal{G}'_1$ and $\mathcal{G}'_2$ such that $|\mathcal{G}'_1| =2 $ and $|\mathcal{G}'_2| =1 $ (including the desired message). For instance, if message $W_1$ is desired, the user can download $W_1$ along with any one of  $\{W_2,W_3,W_4\}$ and any one message form $\{W_5,W_6\}$. Hence, a possible query for retrieving $W_1$ can be $\mathcal{G}'_1= \{1,2\}$ and $\mathcal{G}'_2= \{5\}$. 
The average download cost for this scheme is therefore $D_\mathbf{H}=D_\mathbf{H}(k)=3$, $k\in[1:6]$.

In this example, we observe that the average download cost of LV-PIR is reduced by a constant factor compared to the classical PIR. This factor is equal to the greatest common divisor between the groups sizes, i.e., $\text{gcd}(|\mathcal{G}_1|$,$|\mathcal{G}_2|) =2$. Building upon this motivating example, we next formalize our low-complexity LV-PIR scheme. 
\end{example}


\noindent {\bf A low-complexity LV-PIR scheme for groups.} 
Consider the scenario when the set of messages indices $[1:K]$ can be organized into $P$ non-overlapping groups/subsets  represented by
$\mathcal{G}_1,\dots, \mathcal{G}_P$, such that $\sum_{p=1}^{P}|\mathcal{G}_p|=K$. Two message indices $k_1$ and $k_2$ belong to the same group $\mathcal{G}_p$ if they have identical columns in the matrix $\mathbf{H}$, i.e., $\overrightarrow{\mathbf{H}(:,k_1)}= \overrightarrow{\mathbf{H}(:,k_2)}$. The following Theorem gives an LV-PIR scheme that achieves an upper bound for the optimal download cost. 

\vspace{-5pt}
\begin{theorem}
\label{Thm2}
For an LV-PIR problem with characteristic matrix  $\mathbf{H}$ and  $P$ message groups, $\mathcal{G}_1,\dots, \mathcal{G}_P$, the optimal average download cost is upper bounded by 
\begin{align}
    D^*(\mathbf{H})\leq D^{\text{G}}({\mathbf{H}}) =K/\rho,
\end{align}
where  $\rho=\text{gcd}(|\mathcal{G}_1|,|\mathcal{G}_2|,\dots,|\mathcal{G}_P|)$.
\end{theorem}

\begin{proof}
For ease of illustration and without loss of generality, assume that the first $|\mathcal{G}_1|$ message indices in $[1:K]$ belong to $\mathcal{G}_1$, the following $|\mathcal{G}_2|$ indices belong to $\mathcal{G}_2$, etc. Since matrix $\mathbf{H}$ has only $P$ distinct columns,  let $\mathbf{H}_G$ denote the reduced matrix from $\mathbf{H}$ only containing the $P$ distinct columns in the same order as in $\mathbf{H}$.  From \eqref{eq:S_init}, the prior  distribution of  $S$ (or the RHS of \eqref{eq:Thm1}) can be expressed as
\begin{equation}
\label{eq:S_init_gp}
\begin{split}
  \frac{1}{K}\overrightarrow{\mathbf{H}}\cdot \vec{\mathbbm{1} }=\overrightarrow{\mathbf{H}_G} \cdot \left[\frac{|\mathcal{G}_1|}{K}\ \frac{|\mathcal{G}_2|}{K}\ \cdots\ \frac{|\mathcal{G}_P|}{K}\right]^T.
   \end{split}
\end{equation}
 Now, consider a query $\mathcal{M}_k$ for a desired message $W_k$, where $k \in \mathcal{M}_k$ for correctness. The query $\mathcal{M}_k$ can be described in a general form as $\mathcal{M}_k= \{\mathcal{G}'_1, \ldots, \mathcal{G}'_P\}$, i.e., a subset of messages with indices in the set $\mathcal{G}'_p \subseteq \mathcal{G}_p$ are downloaded from group  $p$, $\forall p\in [1:P]$. For this scheme, the posterior distribution of $S$ given the query (LHS of \eqref{eq:Thm1}) equals
\begin{equation}
\label{eq:S_est_gp}
\frac{1}{|\mathcal{M}_k|}\overrightarrow{\mathbf{H}}\cdot \vec{b}_{k}=\overrightarrow{\mathbf{H}_G}\cdot \left[\frac{|\mathcal{G}'_1|}{|\mathcal{M}_k|}\ \frac{|\mathcal{G}'_2|}{|\mathcal{M}_k|}\ \cdots\ \frac{|\mathcal{G}'_P|}{|\mathcal{M}_k|}\right]^T.
\end{equation}
From \eqref{eq:S_init_gp}, \eqref{eq:S_est_gp} and Lemma~\ref{Thm1}, this scheme satisfies latent-variable privacy if we have
\begin{equation}
 \left[\frac{|\mathcal{G}'_1|}{|\mathcal{M}_k|}\ \frac{|\mathcal{G}'_2|}{|\mathcal{M}_k|}\ \cdots\ \frac{|\mathcal{G}'_P|}{|\mathcal{M}_k|}\right]= \left[\frac{|\mathcal{G}_1|}{K}\ \frac{|\mathcal{G}_2|}{K}\ \cdots\ \frac{|\mathcal{G}_P|}{K}\right].
\end{equation}
To satisfy the above constraint, we must have
\begin{equation}
\label{eq:K/K'ratio}
    \frac{|\mathcal{G}'_p|}{|\mathcal{M}_k|}=\frac{|\mathcal{G}_p|}{K},\quad\forall p\in[1:P],
\end{equation}
which is equivalent to the following condition
\begin{equation}
\label{eq:Kp1/Kp2}
    \frac{|\mathcal{G}'_{p_1}|}{|\mathcal{G}'_{p_2}|}=\frac{|\mathcal{G}_{p_1}|}{|\mathcal{G}_{p_2}|}, \quad \forall p_1,p_2\in[1:P].
\end{equation}
Let $\rho=\text{gcd}(|\mathcal{G}_1|,|\mathcal{G}_2|,\dots,|\mathcal{G}_P|)$. The user can download $|\mathcal{G}'_{p}|=\nicefrac{|\mathcal{G}_p|}{\rho}$ messages from each group $\mathcal{G}_p$ to satisfy  \eqref{eq:Kp1/Kp2}, and consequently \eqref{eq:K/K'ratio}. The download cost for any message $W_k$, which for this scheme is equal to the average download cost, can then be computed as
\begin{equation}
    D^G(\mathbf{H})=\sum_{p=1}^P|\mathcal{G}'_p| =\sum_{p=1}^P\frac{|\mathcal{G}_p|}{\rho}=\frac{K}{\rho}.
\end{equation}
This completes the proof of Theorem \ref{Thm2}. 
\end{proof}

\begin{remark} [On the optimality of the grouping LV-PIR scheme]
We note that depending on the grouping structure of $\mathbf{H}$, the low complexity scheme can  reduce the download cost whenever $\rho>1$. For instance, consider the case when $H$ has all $K$ repeated columns. In this case, $P=1$, and $\rho=K$, i.e., $D^{G}(\mathbf{H})=1$. In other words, we only download the desired message. On the other hand, if $\rho=1$, then, $D^{G}(\mathbf{H})=K$, i.e., the grouping scheme is equivalent to traditional PIR. 
However, this scheme neglects the information in the reduced matrix  $\mathbf{H}_{G}$, which can be further exploited to reduce the download cost. For instance, the exhaustive search scheme can provide better download efficiency as it takes into account the full information of $\mathbf{H}$, as shown in the next example. 
\end{remark}

\begin{example}
\normalfont
Let us consider a LV-PIR setting with $T=3$, $K=4$, and the following characteristic matrix $\mathbf{H}$:
\begin{equation}
\mathbf{H}=\begin{bmatrix}
    0.3 & 0.1&  0.1 & 0.3 \\
    0.4 & 0.2&  0.5 & 0.1 \\
    0.3 & 0.7&  0.4 & 0.6\\
\end{bmatrix}.
\end{equation}
From the structure of $\mathbf{H}$, we  have four groups in total, i.e., $\mathcal{G}_i= \{i\}$, $i=1,2,3,4$ as all the four columns are unique. In this case, we get $\rho=\text{gcd}(1,1,1,1) =1$, then according to Theorem~\ref{Thm2} whenever a message is requested the grouping scheme downloads all $K/\rho=K=4$ messages. However, the exhaustive search scheme gives the following query:
\begin{align}
\label{ex3:query}
    Q^{(\theta)} =\begin{cases}
    \mathcal{L}_1=\{1,2\}, &\quad \theta = 1 \text{ or } \theta=2,\\
        \mathcal{L}_2=\{3,4\}, &\quad \theta = 3 \text{ or } \theta=4,
    \end{cases}
\end{align}
which has a lower average download cost of $D^{\text{ES}}(\mathbf{H}) = 2<4$. 
We can readily verify that the query in \eqref{ex3:query} satisfies the latent-variable privacy constraint in Lemma \ref{Thm1}. The prior distribution of $S$ can be computed as $\frac{1}{K}\mathbf{H}\cdot \vec{\mathbbm{1} }={\mathbf{H}}\cdot [1 \ \ 1 \ \ 1 \ \ 1]^T=  [0.2 \ \ 0.3 \ \ 0.5]^T$. We can check that the posterior distribution of $S$ given any of the two possible queries matches the prior distribution. In particular, if $Q^{(\theta)}=\{1,2\}$, then the posterior distribution can be found as $\frac{1}{2}{\mathbf{H}}\cdot [1 \ \ 1 \ \ 0 \ \ 0]^T=  [0.2 \ \ 0.3 \ \ 0.5]^T$, and similarly, if $Q^{(\theta)}=\{3,4\}$, then the posterior distribution can be found as $\frac{1}{2}{\mathbf{H}}\cdot [0 \ \ 0 \ \ 1 \ \ 1]^T=  [0.2 \ \ 0.3 \ \ 0.5]^T$. 

\end{example}

\section{Discussion}
In this paper, we introduced the problem of latent variable PIR, where information-theoretic privacy is preserved with respect to a latent variable. We formulated the new privacy requirements for the LV-PIR problem and introduced a general scheme that improves the download cost as a function of the characteristic matrix $\mathbf{H}$. We utilized the structure of $\mathbf{H}$ to achieve a low-complexity  LV-PIR construction when messages are partitioned in groups with the same statistical properties with respect to the latent variable.  An immediate open problem is to obtain lower bounds and characterize the capacity of LV-PIR. Moreover, extensions of the LV-PIR model to multiple databases and multiple latent variables are other possible interesting directions. 

 \newpage
\bibliographystyle
{IEEEtran}
{  \bibliography{IEEEabrv,PIR_paper}}


\end{document}